\documentclass[12pt]{amsart}
\usepackage[]{color}
\usepackage{amssymb}
\usepackage[]{graphicx}
\usepackage{setspace}

\def\R{\mathbb{R}}
\def\C{\mathbb{K}}
\def\K{\mathbb{K}}
\def\IS{\mathbb{IS}}
\def\IG{\mathbb{IG}}
\def\sgn{\mathop{\hbox{sgn}}}
\def\eps{\varepsilon}
\def\S{\mathcal S}
\def\F{\mathcal F}
\providecommand{\abs}[1]{\lvert#1\rvert}
\usepackage{hyperref}
\definecolor{myblue}{RGB}{25,25,112}

\hypersetup{
    colorlinks=true,
    linkcolor=blue,
    filecolor=magenta,      
    urlcolor=myblue,
    citecolor=blue,
}

\newtheorem{thm}{Theorem}
\newtheorem{lemma}{Lemma}
\newtheorem{cor}{Corollary}
\newtheorem{prop}{Proposition}
\newtheorem{remark}{Remark}
\usepackage{natbib}

\newcommand{\bb}{\color{black}}
\newcommand{\ee}{\color{black}}

\begin{document}
\bibliographystyle{plainnat}

\title[Permanence via invasion graphs]{Permanence via invasion graphs: \\Incorporating community assembly into\\ Modern Coexistence Theory}
\author[J. Hofbauer and S.J. Schreiber]{Josef Hofbauer and Sebastian J. Schreiber$^\dagger$}
\address{Department of Mathematics, University of Vienna}
\email{josef.hofbauer@univie.ac.at}
\address{Department of Evolution and Ecology, University of California, Davis }
\email{sschreiber@ucdavis.edu}
\footnotetext[2]{Corresponding author: sschreiber@ucdavis.edu}
\begin{abstract}  
To understand the mechanisms underlying species coexistence, ecologists often study invasion growth rates of theoretical and data-driven models. These  growth rates correspond to average per-capita growth rates of one species with respect to an ergodic measure supporting other species. In the ecological literature, coexistence often is equated with the invasion growth rates being positive. Intuitively, positive invasion growth rates ensure that species recover from being rare. To provide a mathematically rigorous framework for this approach, we prove theorems that answer two questions: (i) When do the signs of the invasion growth rates determine coexistence? (ii) When signs are sufficient, which invasion growth rates need to be positive? We focus on deterministic models and equate coexistence with permanence, i.e., a global attractor bounded away from extinction. For models satisfying certain technical assumptions, we introduce invasion graphs where vertices correspond to proper subsets of species (communities) supporting an ergodic measure and directed edges correspond to potential transitions between communities due to invasions by missing species. These directed edges are determined by the signs of invasion growth rates. When the invasion graph is acyclic (i.e. there is no sequence of invasions starting and ending at the same community), we show that permanence is determined by the signs of the invasion growth rates. In this case, permanence  is characterized by the invasibility of all $-i$ communities, i.e., communities without species $i$ where all other missing species having negative invasion growth rates. To illustrate the applicability of the results, we show that dissipative Lotka-Volterra models generically satisfy our technical assumptions and  computing their invasion graphs reduces to solving systems of linear equations. We also apply our results to models of competing species with pulsed resources or sharing a predator that exhibits switching behavior. Open problems for both deterministic and stochastic models are discussed. Our results highlight the importance of using concepts about community assembly to study coexistence. 
\end{abstract}
\maketitle

\section{Introduction}

Understanding the mechanisms allowing interacting populations to co-occur underlies many questions in ecology, evolution, and epidemiology: When are species limited by a common predator able to coexist? What maintains genetic diversity within a species? Why do multiple pathogen strains persist in host populations? One widely used metric for understanding coexistence is invasion growth rates: the average per-capita growth rates of populations when rare. This approach has a long history going back to the work of \citet{macarthur_levins1967,roughgarden1974,chesson-78} and \citet{turelli-78a}. This earlier work focused on the case of two competing species and lead to the mutual invasibility condition for coexistence. Namely, if each competitor has a positive invasion growth rate when the other species is at \bb stationarity\ee, then the competitors coexist.  \bb Mathematically, this form of coexistence corresponds to all species densities tending away from extinction. This corresponds to permanence or uniform persistence for deterministic models~\citep{schuster-etal-79,sigmund1984permanence,butler-freedman-waltman-86,garay-89,hofbauer-so-89,hutson-schmitt-92} and stochastic persistence for stochastic models~\citep{chesson-82,chesson-ellner-89,jmb-11}. \ee

A key feature of the mutual invasibility criterion is that coexistence is determined by the signs of invasion growth rates. For many classes of multispecies models, positive invasion growth rates of  at least one missing species from each subcommunity is a necessary condition for  permanence to persist under small structural perturbations, i.e., robust permanence~\citep{hutson-schmitt-92,jde-00}. However, it need not be sufficient as in the case of three competing species exhibiting a rock-paper-scissor dynamic~\citep{may-leonard-75}. In this case, all single species equilibria can be invaded by a missing species but coexistence depends on quantitative information about the  invasion growth rates at these equilibria~\citep{hofbauer-81,hofbauer-sigmund-98,jde-00,jde-10}.  This raises the question, when is it sufficient to know the sign structure of the invasion growth rates? Are rock-paper-scissor type dynamics the main barrier to qualitative conditions for permanence? 

Invasion growth rates are the basis of what has become known as modern coexistence theory (MCT) or Chesson's coexistence theory~\citep{chesson-94,letten2017,barabas2018,Chesson_2018,ellner_snyder2018,grainger_levine2019,grainger2019,godwin2020,chesson2020}. \bb To understand the mechanisms underlying coexistence, positive invasion growth rates \ee are decomposed into biologically meaningful components and compared to the corresponding components for the resident species~\citep{chesson-94,ellner_snyder2018}. This decomposition allows ecologists to identify which coexistence mechanisms may or may not be operating in their system~\citep{chesson-94,adler-etal-10,ellner-etal-16,ellner_snyder2018}. \bb For many applications of MCT, simple variants of the mutual invasibility condition  determine which invasion growth rates  need to be positive for coexistence~\citep{chesson-kuang-08,Chesson_2018}.  For example, for two prey species sharing a common resource and  a common predator, coexistence is determined by the invasion growth rates of each  prey species into the three species community determined by its absence~\citep{chesson-kuang-08}. However, for more complex models, it is less clear how invasion growth rates determine coexistence~\citep{barabas2018,Chesson_2018}. \ee

%In the words of \citep{macarthur_levins1967} ``such a community can retain all $n$ species if any one of them can increase when rare, i.e., when [species $i$'s density] $X_i$ is near zero and all the others are at the equilibrium values which they would reach in the absence of $X_i$.''

% ``the mutual invasion criterion is met when all species in a community have positive invasion growth rates''

Here, we address these issues by introducing a mathematically precise notion of the  \emph{invasion graph} that describes all potential transitions between subcommunities via invasions (Section \ref{sec2}). For our models, we focus on a specific set of generalized ecological equations~\citep{jmb-18}. However, the proofs for our main results should hold for more general classes of models (e.g. reaction-diffusion equations~\citep{zhao-03}, discrete-time models~\citep{garay-hofbauer-03,roth-etal-17}) where conditions for permanence are determined by Morse-decompositions~\citep{conley-78} and invasion growth rates.  Under suitable assumptions about the ecological dynamics described in Section~\ref{subsec:assumptions}, we prove that if the invasion graph is acyclic and each  subcommunity is invadable, then the species coexist in the sense of robust permanence (Section \ref{sec3}). In fact, we show that invasibility only needs to be checked at  $-i$ communities, i.e., communities without species $i$ that are uninvadable by the remaining missing species. This sufficient condition always is a necessary condition for robust permanence. We show that our  assumptions in Section~\ref{subsec:assumptions} hold generically for Lotka-Volterra systems  and provide a simple algorithm for computing invasion graphs (Section \ref{sec4}). We also apply our results to models of two prey sharing a switching predator and a periodically-forced chemostat with three competing species (Section \ref{sec4}). We conclude with a discussion about open problems and future challenges (Section \ref{sec5}). 

\section{Ecological equations, invasion schemes, and invasion graphs~\label{sec2}}

\subsection{Models, assumptions, and permanence.} To cover ecological models accounting for species interactions, population structure (e.g. spatial, age, or genotypic), and auxiliary (e.g. seasonal forcing or abiotic variables), we consider a class of  ordinary differential equations introduced by \citet{jmb-18}. In these equations, there are $n$ interacting species with densities $x=(x_1,x_2,\dots,x_n)$ taking values in the non-negative cone of $\R^n$. In addition, there are auxiliary variables, $y=(y_1,y_2,\dots,y_m)$, taking values in a compact subset $Y$ of $\R^m.$ These auxiliary variables may describe internal feedbacks within species (e.g. genetic or spatial structure) or external feedbacks (e.g. environmental forcing or abiotic feedback variables). \bb  Let $z=(x,y)$ denote the state of the system. \ee In this framework, the dynamic of species $i$ is determined by its per-capita growth rate $f_i(z)$, while the dynamics of the auxiliary variables are determined by some multivariate function $g(z)=(g_1(z),\dots, g_m(z))$. Thus, the equations of motion are
\begin{equation}\label{eq:main}
\begin{aligned}
\frac{dx_i}{dt}&=x_i f_i(z)\quad\ \bb\mbox{ for }\ee i \in [n]:=\{1,2,\dots,n\}\\
\frac{dy}{dt}&=g(z) \quad \bb\mbox{where } z=(x,y).\ee
\end{aligned}
\end{equation}
The state-space for these dynamics is the non-negative orthant $\C=[0,\infty)^n \times Y$. 
The boundary of this orthant, $\C_0=\{(x,y)\in \C: \prod_i x_i=0\}$, corresponds to the extinction of one or more species. The interior of this orthant, $\C_+=\C\setminus \C_0$, corresponds to all species being present in the system. \bb For any initial condition $z\in \C$ at time $t=0$, we let $z.t$ denote the solution to \eqref{eq:main} at time $t$.\ee 

Our first two standing assumptions for the equations \eqref{eq:main} are:
\begin{description}
\item[A1] The functions $(x,y)\mapsto x_i f_i(x,y)$ and $(x,y)\mapsto g(x,y)$ are locally Lipschitz and, consequently, there exist unique solutions $z.t$ to \eqref{eq:main} for any initial condition $z=(x,y) \in \C$. 
\item[A2] The system is \emph{dissipative}: There exists a compact attractor $\Gamma\subset \C$ such that $\mbox{dist}(z.t,\Gamma)\to 0$ as $t\to +\infty$ for all $z=(x,y)\in\K$. Let $\Gamma_0=\Gamma\cap \C_0$.
\end{description}

We are interested in when the equations~\eqref{eq:main} are robustly permanent, i.e., species persist following large perturbations of their initial conditions \bb and \ee small perturbations of the equations governing their dynamics~\citep{hutson-schmitt-92,jde-00,garay-hofbauer-03,jmb-18}. \eqref{eq:main} is \emph{permanent} if there exists a compact set $K\subset \K \setminus \K_0$ such that for all $z=(x,y)\in \K\setminus \K_0$, $z.t\in K$ for $t$ sufficiently large. Namely, for all initial conditions supporting all species, the species densities are eventually uniformly bounded away from the extinction set $\K_0.$ \eqref{eq:main} is \emph{robustly permanent} if it remains permanent under perturbations of $f_i$ and $g$ that satisfy assumptions \textbf{A1}--\textbf{A2}. More precisely, given any compact neighborhood $V$ of $\Gamma$, there exists $\delta>0$ such that $\frac{dx_i}{dt}=x_i \tilde f_i (x,y), \frac{dy}{dt}=\tilde g(x,y)$ is permanent whenever $\| (f(z),g(z))-(\tilde f(z),\tilde g(z))\|\le \delta$ for all $z=(x,y)\in V$ and $(\tilde f,\tilde g)$ satisfy assumptions \textbf{A1}--\textbf{A2} with a global attractor $\tilde \Gamma$ contained in $V$.

\subsection{Invasion growth rates, schemes, and graphs}\label{subsec:assumptions}

To understand whether species coexist in the sense of permanence, we have to consider species per-capita growth rates when rare, i.e., invasion growth rates. These are best described using ergodic probability measures that correspond to indecomposible dynamical behaviors of the model. Recall, a Borel probability measure $\mu$ on $\K$ is \emph{invariant} for \eqref{eq:main} if $\int h(z)\mu(dz)=\int h(z.t)\mu(dz)$ for any continuous function $h:\K\to \R$ and any time $t$. Namely, the expected value of an ``observable'' $h$ does not change in time when the initial condition \bb is \ee chosen randomly with respect to $\mu$. An invariant probability measure $\mu$ is \emph{ergodic} if it can not be written as a non-trivial convex combination of two invariant probability measures, i.e., if $\mu=\alpha \mu_1+(1-\alpha)\mu_2$ for two distinct invariant measures $\mu_1,\mu_2$, then $\alpha=1$ or $\alpha=0$. The simplest example of an ergodic probability measure is a Dirac measure $\mu=\delta_{z^*}$ associated with an equilibrium $z^*$ of \eqref{eq:main}. This Dirac measure is characterized by $\int h(z)\mu(dz)=h(z^*)$ for every continuous function $h:\K\to\R$. Alternatively, if $z^*.t$ is a periodic solution with period $T$, then the measure $\mu$ defined by averaging along this periodic orbit is an ergodic measure\bb~\citep{mane-83,jde-00}\ee. Specifically, $\int h(z)\mu(dz)=\frac{1}{T}\int_0^T h(z^*.t)dt$ for all continuous $h:\K\to\R.$ \bb More generally, the ergodic theorem 
%\citet[Theorem 6.1]{mane-83}'s proof of the ergodic decomposition theorem 
implies that for every ergodic measure $\mu$ there exists an initial condition $z^*$ such that $\mu$ is determined by averaging along the orbit of 
$z^*$, i.e., $\int h(z)\mu(dz)=\lim_{T\to\infty} \frac{1}{T}\int_0^T h(z^*.t)dt$ for all continuous $h:\K\to\R$. \ee

For any subset of species $S\subset [n]=\{1,2,\dots,n\}$, we define \[\F(S) :=\{(x,y)\in \K:x_j>0 \ \mbox{if and only if} \ j\in S\}\] to be \emph{the open face of $\K$ supporting the species in $S$}. For an ergodic measure $\mu$, we define \emph{the species support $S(\mu) \subset [n]$ of $\mu$} to be  the smallest subset of $[n]$  such that  $\mu(\F(S(\mu)) )= 1$. 

To understand whether a missing species $i\notin S(\mu)$ not supported by an ergodic measure $\mu$ can increase or not, we introduce the non-autonomous, linear differential equation
\[
\frac{d\widetilde x_i}{dt}=\widetilde x_i f_i(z.t)
\]
to approximate the dynamics of species $i$'s density $\widetilde x_i$ when introduced at small densities. The solution of this linear differential equation satisfies 
\[
\log  \frac{\widetilde x_i(t)}{\widetilde x_i(0)}=\int_0^t f_i(z.s)ds.
\]
Birkhoff's Ergodic Theorem implies that   
\[
\lim_{t\to\infty}\frac{1}{t}\int_0^t f_i(z.s)ds=\int f_i(z)\mu(dz) \ \mbox{ for $\mu$ almost every initial condition $z=(x,y)$}.
\]
Consequently, we define the \emph{invasion growth rate of species $i$ at $\mu$} as
\[
r_i(\mu) :=\int_{\C} f_i(z)\,\mu(dz).
\]
$r_i(\mu)$ is also defined for the resident species $i\in S(\mu)$ supported by $\mu$. In this case,  we don't interpret $r_i(\mu)$ as an invasion growth rate. Indeed, the following lemma shows that $r_i(\mu)=0$ in this case, i.e., resident species have a zero invasion growth rate.

\begin{lemma}
Let $\mu$ be an ergodic probability measure for \eqref{thm:main}. Then 
$r_i(\mu) = 0$ for all $i \in S(\mu)$. 
\end{lemma}

The proof of this lemma follows from the argument given for  models without auxiliary variables $y$ found in \citep[Lemma 5.1]{jde-00}. 

\begin{proof} Let $i\in S(\mu)$ be given. Let $\pi_i:\K\to\R$ be the projection onto the $i$-th component of the $x$ coordinate, i.e., $\pi_i(z)=x_i$ when $z=(x,y)\in \K$.  Since $\mu(\F(S(\mu))=1$, Birkhoff's Ergodic Theorem implies that there exists an invariant Borel set $U\subseteq  \F(S(\mu))$ such that $\mu(U)=1$ and
\begin{equation}\label{eq:bet}
\lim_{t\to\infty}\frac{1}{t}\int_0^t f_i(z.s)ds=r_i(\mu)
\end{equation}
whenever $z\in U$. Choose an open set $V$ such
that its closure $\overline {V}$ is contained in $\F(S(\mu))$, $\overline{V}$ is compact, and $\mu(V\cap
U)>0$. By the Poincar\'{e} recurrence theorem, there exists $z\in
V\cap U$ and an increasing sequence of real numbers $t_k\uparrow \infty$ such that $z.t_k\in V$ for all $k\ge 1$. Since is $\overline{V}$ is compact, there exists a
$\delta>0$ such that
    \begin{equation}\label{eq:bounds}
    1/\delta\le \pi_i(z.t_k) \le \delta
    \end{equation}
for all $k$. As 
$\log \frac{\pi_i(z.t)}{\pi_i(z)}=\int_0^t f_i(z.s)ds$,  \eqref{eq:bet} and \eqref{eq:bounds}  imply that
    \[
    r_i(\mu) = \lim_{t\to\infty}\frac{1}{t}\int_0^t f_i(z.s)ds= 
    \lim_{k\to\infty} \frac{1}{t_k}\log \frac{\pi_i(z.t_k)}{\pi_i(z)} =0.
   \]
\end{proof}

We make the following additional standing assumption: 
\begin{description}
	\item[A3a] For each ergodic invariant Borel probability measure $\mu$ supported by $\Gamma_0$, $r_j(\mu)\neq 0$ for all $j\notin S(\mu)$, and 
	\item[A3b] $\sgn r_j(\mu) = \sgn r_j(\nu)$ for any two ergodic measures $\mu, \nu$ with $S(\mu) = S(\nu)$, and all $j$.
\end{description}

Assumption \textbf{A3a} requires the invasion growth rates $r_i(\mu)$ are non-zero for species not supported by $\mu$. This assumption holds typically for dissipative Lotka-Volterra systems or systems with a finite number of ergodic measures. Due to their time averaging property, assumption \textbf{A3b} holds for all Lotka-Volterra systems and replicator equations~\citep{hofbauer-sigmund-98} and certain types of periodically forced versions of these equations~\citep{jmb-18}. This assumption automatically holds when each face supports at most one invariant probability measure (e.g. there is a unique equilibrium, periodic orbit, or quasi-periodic motion in a given face). Sometimes this sign parity also can be verified when the per-capita growth functions $f_i$ exhibit the right convexity properties (e.g. \citep{kon-04,jde-04}). \bb For non-Lotka Volterra systems, it is possible for the per-capita growth rates of a missing species to have opposite signs at different ergodic measures. In this case, assumption \text{A3b} fails.  For example, this failure arises in models of two predator species competing for a single prey species~\citep{mcgehee-armstrong-77}. If one predator has a type II functional response, then the predator-prey subsystem may simultaneously have an unstable equilibrium (defining one ergodic measure) and a stable limit cycle (defining another ergodic measure). \citet{mcgehee-armstrong-77} showed that the invasion growth rates of the other predator species may be positive at the stable limit cycle but negative at the unstable equilibrium. A similar phenomena arises in models  of  two prey species sharing a common predator~\citep{jde-04}. \ee

In light of assumption \textbf{A3}, we can uniquely define 
\[
r_i(S) 	= \sgn r_i(\mu) \mbox{ if }S = S(\mu).
\]
for each subset $S\subset [n]$ of species.  Let \emph{$\S$ be the set of all subcommunities}: all proper subsets $S$ of $[n]$ such that $S = S(\mu)$ for some ergodic measure $\mu$. For $S\in \S$, there are at most $n-1$ species and, consequently, $\mu(\C_0)=1$ for any ergodic measure $\mu$ supported by $S$. Furthermore, $\S$ isn't empty as it always contains the empty community $\emptyset\subset [n]$. Let $|\S|$ be the number of elements in $\S$, i.e., the number of subcommunities.

Now, we introduce our two main definitions. We define the \textbf{invasion scheme $\IS$} to be the table of the signs of invasion growth rates $\{(r_i(S))_{ i\in [n]}: S\in \S\}$. When viewed as a $|\S| \times n$ matrix, the invasion scheme is the signed version of the characteristic matrix introduced in \citep{hofbauer-94}. The rows of this matrix correspond to the different subcommunities while the columns correspond to the different species. We define the \textbf{invasion graph} $\IG$ as the directed graph with vertex set $\S$ and  a directed edge from $S\in \S$ to $T\in \S$ if 
 
\begin{itemize}
    \item $S \not= T$,
	\item $r_j(S)>0$ for all $j \in T\setminus S$, and 
	\item $r_i(T)<0$ for all  $i \in S \setminus T$.
\end{itemize}
The first condition implies that there are no self-loops in the invasion graph. The second condition implies that all the species in $T$ missing from $S$ can invade $S$. The third condition \bb allows for the loss of species from $S$ that are not in $T$ and ensures that these lost species can not invade $T$. \ee One can view the invasion graph as describing all potential transitions from one subcommunity $S\in \S$ to another subcommunity $T\in \S$ due to invasions of missing species.

\begin{remark} If $z\in \K$ is such that its $\alpha$-limit set lies in $\F(S)$ for a proper subset $S\subset [n]$ and its $\omega$-limit lies in $\F(T)$ for another proper set $T\neq S$, then there is a directed edge from $S$ to $T$. The proof follows from the arguments presented in \ref{sec:proof1}.\end{remark}

\bb\begin{remark} As we focus on determining whether or not $[n]$ is permanent, the invasion graph $\IG$ doesn't include $[n]$. However, for visualization purposes, we include $[n]$ in our plots whenever $[n]$ is permanent (see, e.g., Figure~\ref{fig:LV})\end{remark}\ee

\section{Main Results~\label{sec3}}
The following theorem  partially answers our main question,``when is knowing only qualitative information of the invasion growth rates $r_i(\mu)$ (namely, their sign) sufficient for determining robust permanence?'' Recall, that a directed graph is acyclic if there is no path of directed edges starting and ending at the same vertex, i.e., there are no cycles.  

\begin{thm}\label{thm:main}
Assume that \textbf{A1}--\textbf{A3} hold and $\IG$ is acyclic. Then  \eqref{eq:main} is robustly permanent if for each $S \in \S$
there is  $i$ such that $r_i(S)>0$ i.e., each subcommunity is invadable.
\end{thm}

\bb\begin{remark} If $(f,g)$ are twice continuously differentiable and there exists $S\in \S$ such that $r_i(S)<0$ for all $i\in [n]\setminus S$, then Pesin's Stable Manifold Theorem (see, e.g., \citet{pugh-shub-89}) implies there exists $z\in \C_+ $ such that $\omega (z)\subset \C_0$ and, consequently, \eqref{eq:main} is not permanent.\end{remark}\ee

The proof of Theorem~\ref{thm:main} is given in Appendix A. The  idea of the proof is as follows. The invasion graph being acyclic allows us to construct a Morse decomposition~\citep{conley-78}  of the flow on  the extinction set $\C_0$. Each component of the Morse decomposition corresponds to a subcommunity in the invasion graph. The invasibility conditions ensure that the stable set of each component of the Morse decomposition doesn't intersect the non-extinction set $\C\setminus \C_0$~\citep{jde-00,garay-hofbauer-03,jmb-18}. Then one can apply classic results about permanence~\citep{butler-freedman-waltman-86,garay-89,hofbauer-so-89}.

We derive two useful corollaries from this theorem. First, we show that if the invasion graph is acyclic, then checking for robust permanence  only requires checking the invasibility conditions for a special subset of  subcommunities. Specifically, given a species $i$, a community $S\subset[n]\setminus\{i\}$ is a \emph{$-i$ community} if $r_j(S)\le 0$ for all $j\neq i$. \bb In words, a $-i$ community is a community $S$ that doesn't include species $i$ and that can not be invaded by any of the missing species $j\notin S$ except possibly species $i$. \ee

\begin{cor}\label{cor1}
Assume that \textbf{A1}--\textbf{A3} hold and $\IG$ is acyclic. Then  \eqref{eq:main} is robustly permanent if $r_i(S)>0$ for each $-i$ community $S$ and $i\in [n]$. 
\end{cor}

\bb 
\begin{proof} Let $S\in \S$ and $i\in [n]\setminus S$. By assumption \textbf{A3}, either $r_j(S)<0$ for all $j\in [n]\setminus S\cup\{i\}$ (i.e. $S$ is a $-i$ community) or $r_j(S)>0$ for some $j\in [n]\setminus S\cup\{i\}$ (i.e. $S$ is not a $-i$ community). As, by assumption, $r_i(S)>0$ for the $-i$ communities, applying Theorem~\ref{thm:main} completes the proof. 
\end{proof}
\ee

\begin{remark}\label{remark2} In general, $-i$  communities correspond to subcommunities $S$ for which 
there is an initial condition $z=(x,y)$ such that (i) $x_i=0$ and $x_j>0$ for $j\neq i$, and (ii) the $\omega$-limit set of $z$ intersects $\F(S)$. Indeed, if there is a community $S \subset[n]\setminus\{i\}$ and initial condition $z$ satisfying (i) and (ii), then the proof of Lemma 2 in Appendix A implies $S$ is a $-i$ community. Conversely, if $S$ is a $-i$ community and the functions $f_i,g$ are twice continuously differentiable, then Pesin's Stable Manifold Theorem (see, e.g., \citet{pugh-shub-89}) implies there exists an initial condition $z$ satisfying (i) and (ii).
\end{remark}

Our second corollary concerns \bb average Lyapunov function condition for permanence due to \citet{hofbauer-81}. This sufficient condition is  
\begin{description}
\item [H] There exist positive constants $p_1,\dots,p_n$ such that $\sum_i p_i r_i(\mu)>0$ for all ergodic measures $\mu$ with $S(\mu)\in \S$.
\end{description}
One can ask ``when does knowing only the signs of the $r_i(\mu)$ ensure that condition \textbf{H} holds?'' To answer this question, \ee  we say the invasion scheme $\IS$ is \emph{sequentially permanent} if there is an ordering of the $n$ species, say $\ell_1,\ell_2,\dots,\ell_n$, such that column $\ell_i$ of $\IS_i$ has only non-negative entries where $\IS_1=\IS$ and $\IS_i$ for $i\ge 2$ is defined by removing the rows of $\IS_{i-1}$ where species $\ell_{i-1}$ has a positive per-capita growth rate. As sequential permanence  implies that the invasion graph is acyclic \bb and $\max_i r_i(S)>0$ for all $S\in \S$, we get the following result:

\begin{cor}\label{cor:IS}Assume that \textbf{A1}--\textbf{A3} hold.
If $\IS$ is sequentially permanent, then \eqref{eq:main} is robustly permanent. 
\end{cor}

\bb If $\IG$ is sequentially permanent, then the invasion schemes of \eqref{eq:main} restricted to each of the communities $\{\ell_1\}, \{\ell_1,\ell_2\}, \dots, [n]$ are also sequentially permanent. Hence, each of the communities in this sequence is also permanent. 
Such a sequential way to prove permanence  has been used by \citet[p.~877 ff]{hofbauer-kon-saito-08}. A simple example of a system which is not sequentially permanent but has an acyclic
invasion graph is given in section \ref{sec:switch-pred}.\ee

\bb To see how sequential permanence relates to condition \textbf{H}, consider the special case where $r_i(\mu)=r_i(\nu)$ for ergodic measures satisfying $S(\mu)=S(\nu)$ e.g. a Lotka Volterra model or  models where each face supports at most one ergodic measure. For each $S\in \S$ and $i\in [n]$, let $C_i(S)=r_i(\mu)$ where $S(\mu)=S$. This characteristic matrix $C=\{C_i(S)\}_{S,i}$  has the same sign pattern as the invasion scheme $\IS=\{r_i(S)\}_{S,i}$. For a vector $p=(p_1,\dots,p_n)$, we write $p\gg 0$ if $p_i>0$ for all $i$. Condition \textbf{H} is equivalent to $Cp\gg 0$ for some $p=(p_1,\dots,p_n)\gg0.$ The following algebraic proposition shows that \textbf{H} is guaranteed by the sign structure of the $C$ if and only if $\IS$ is sequentially permanent.  The proof  is  in ~\ref{sec:proof2}.

\begin{prop}\label{prop:IS}
Let $\IS$ be an invasion scheme. Then $\IS$ is sequentially permanent if and only if for every matrix $C$ with $\sgn(C)=\IS$, $Cp\gg 0$ for some $p=(p_1,\dots,p_n)\gg 0.$
\end{prop}
\ee

\section{Applications~\label{sec4}}

To illustrate the use of our results, we first describe how to verify them for Lotka-Volterra systems and also illustrate their application to two non Lotka-Volterra systems: a model of  competing prey sharing a switching predator, and a periodically-forced model of three competing species. For the first two applications, the conditions of Theorem~\ref{thm:main} are evaluated analytically, while in third application, we verify the conditions numerically. 

\begin{figure}
\includegraphics[width=0.4\textwidth]{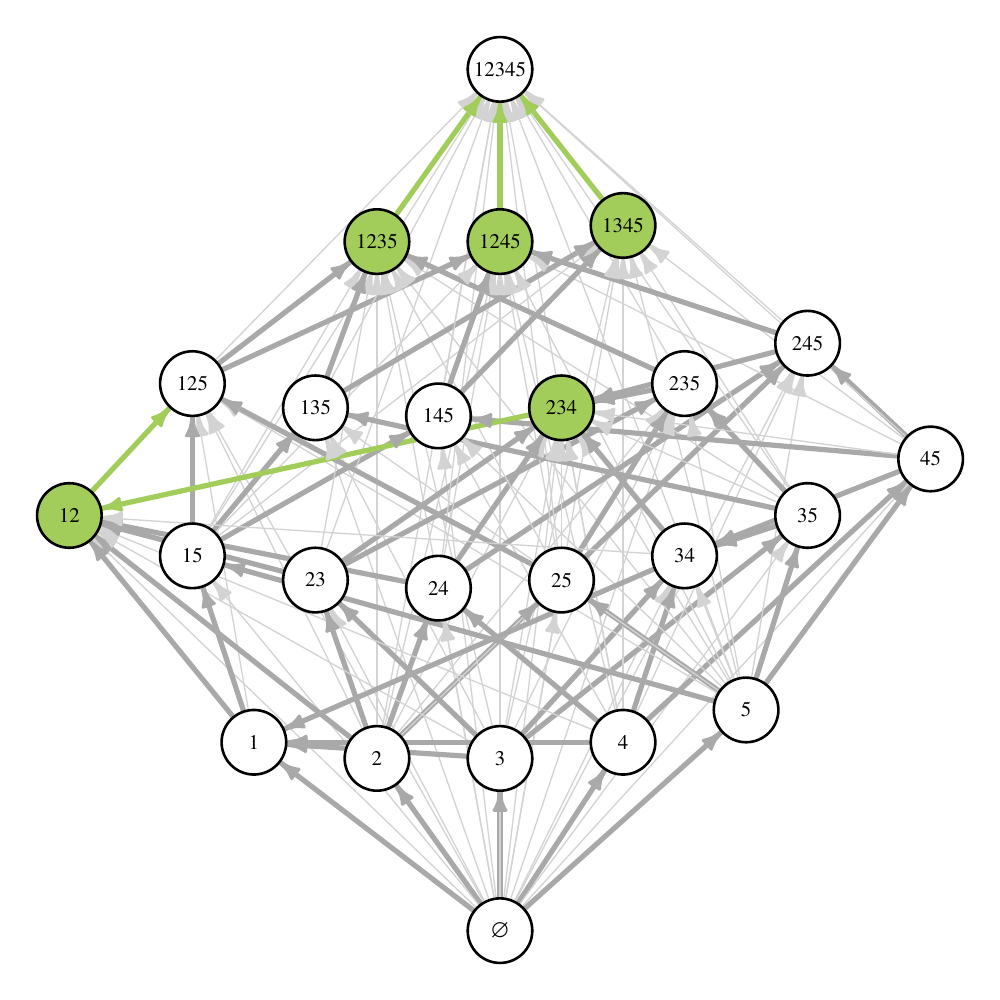}\includegraphics[width=0.4\textwidth]{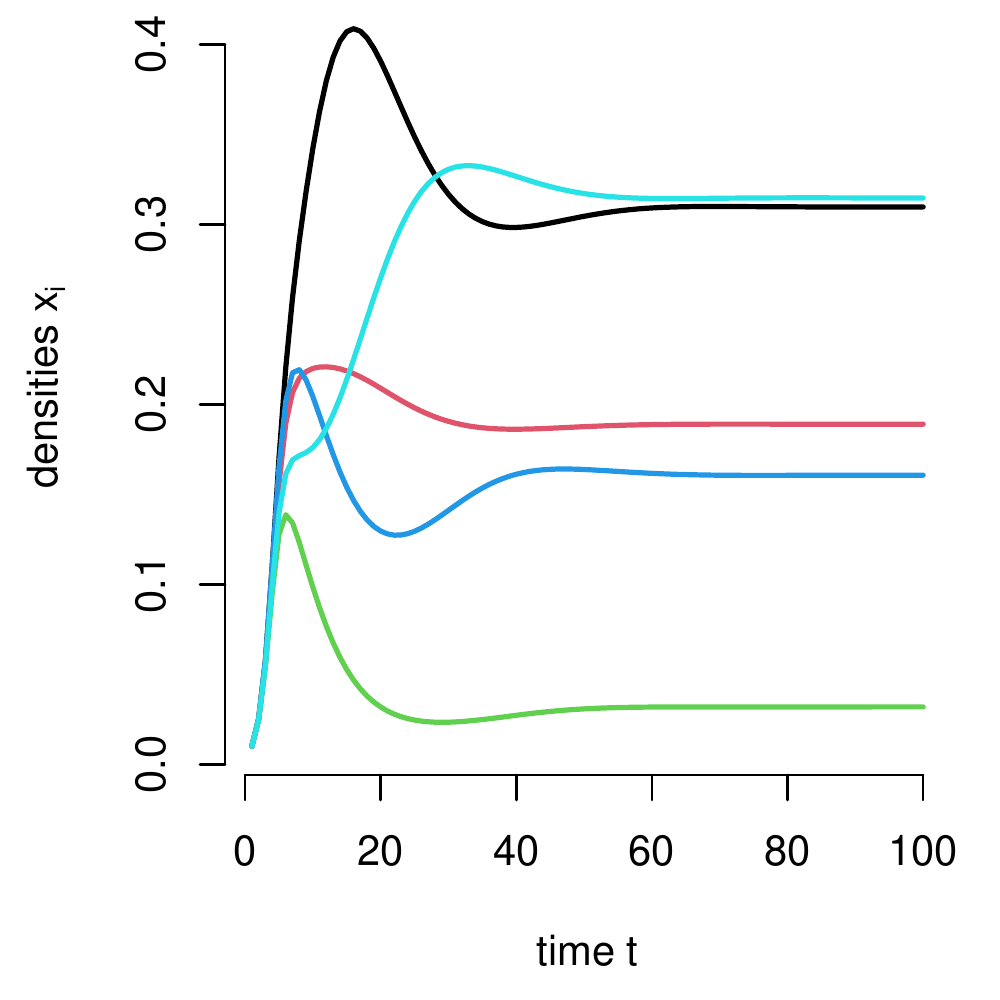}
\includegraphics[width=0.4\textwidth]{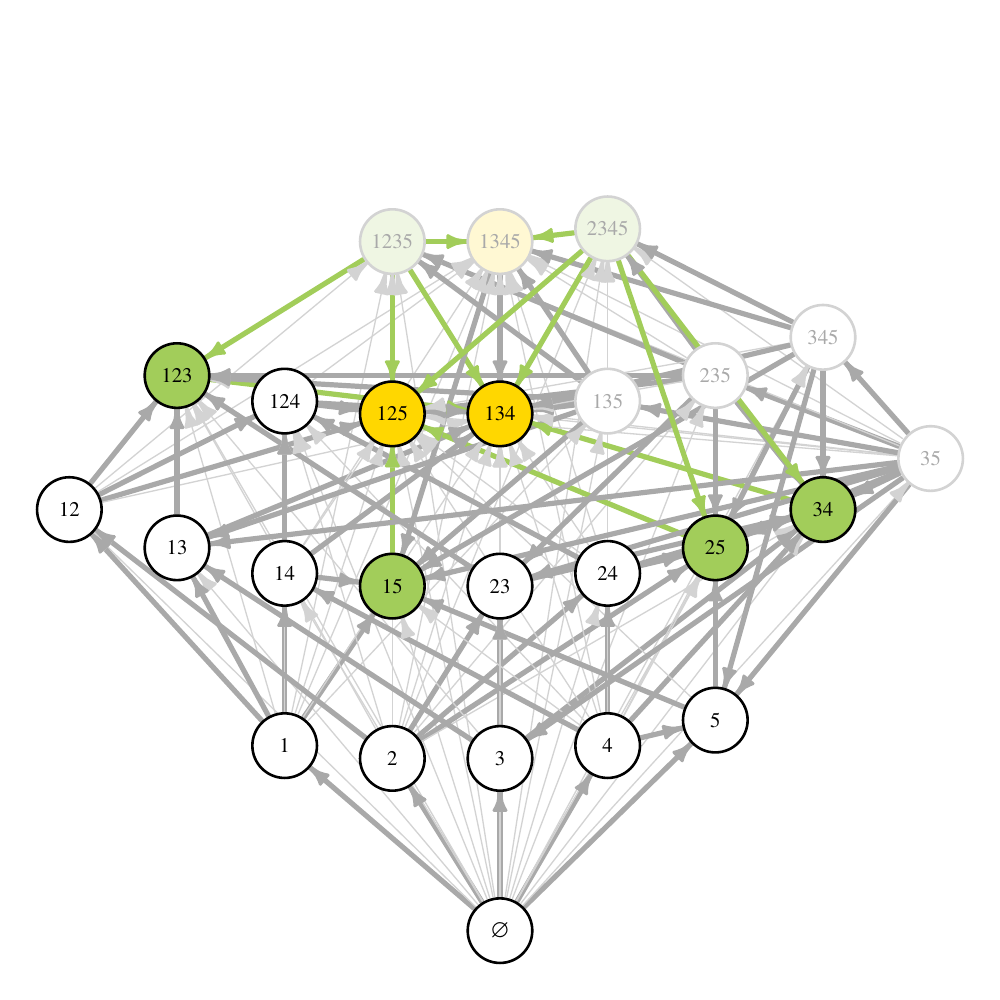}\includegraphics[width=0.4\textwidth]{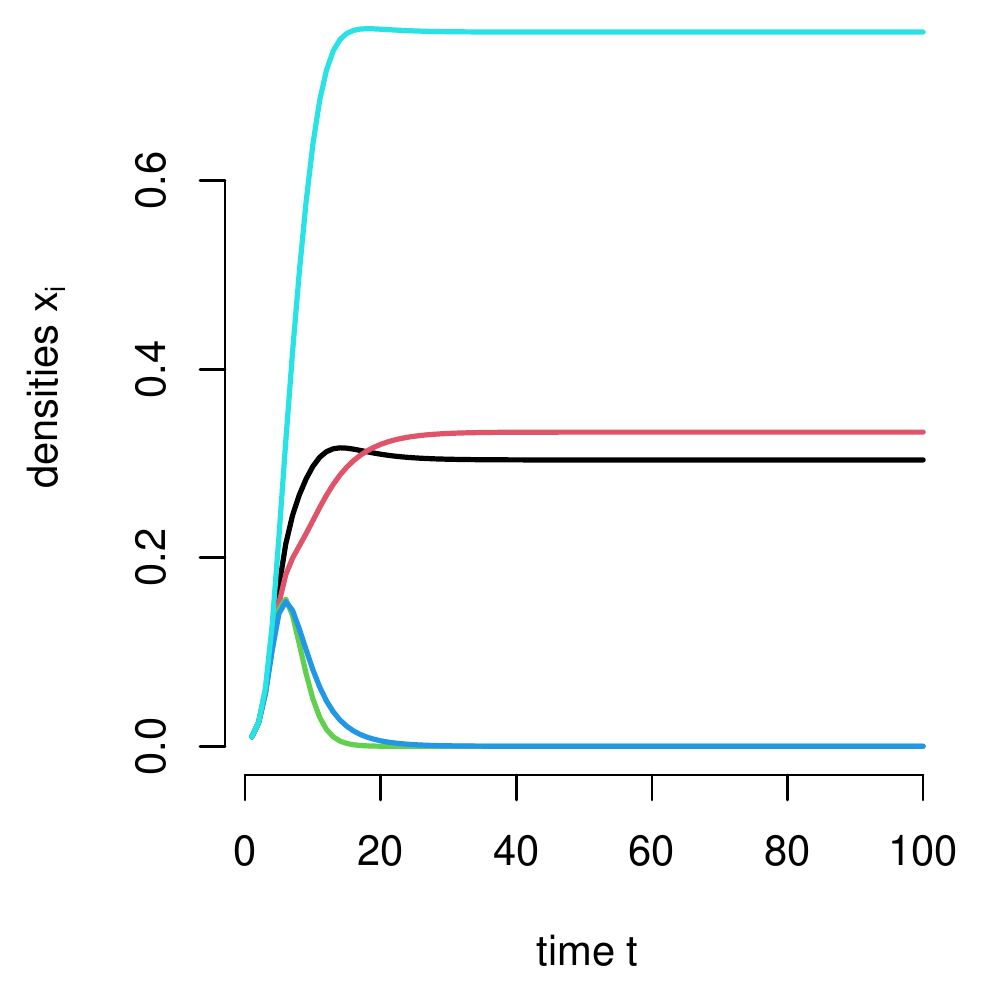}
\caption{Invasion graphs,  $-i$ communities, and sample simulations for two  $5$ species competitive Lotka-Volterra models.  Vertex labels correspond to the species in the community. $-i$ communities for which species $i$ has a positive invasion growth rate are colored olive green, and otherwise gold. \bb Lighter shaded vertices correspond to non-permanent communities, all others correspond to permanent communities. Thicker directed edges correspond to single species invasions. Green directed edges indicate transitions due to species $i$ invading a $-i$ community. \ee  Both models have acyclic invasion graphs, but only the model in the top panels allows for robust permanence. Sample simulations of the models are shown in the right hand panels. Parameter values in \ref{sec:parms}.}\label{fig:LV}
\end{figure}

\subsection{Lotka-Volterra systems\label{sec:LV} }

Consider the Lotka-Volterra equations where $x$ is the vector of species densities and there are no $y$ variables (see, however, below for several extensions involving auxiliary variables). Let $A$ be the $n\times n$ matrix corresponding to the species interaction coefficients and $b$ the $n\times 1$ vector of intrinsic rates of growth. Then $f(x)=Ax+b$. 

Assume $A$ and $b$ are such that the system is dissipative (i.e. \textbf{A2} holds). \citet[ch.~15.2]{hofbauer-sigmund-98} provide various algebraic conditions that ensure dissipativeness. Furthermore, assume that each face of the non-negative orthant has at most one internal equilibrium. Under these assumptions, the Lotka-Volterra system exhibits the time averaging property. Namely, if $z=x$ is an initial condition such that the $\omega$-limit set of $x.t$ is contained in $\F(S)$ for some $S\subset \{1,\dots,n\}$, then 
\[
\lim_{t\to \infty} \frac{1}{t}\int_0^t x.s\,ds=x^*
\]
where $x^*$ is the unique equilibrium in $\F(S)$.  The invasion growth rate of species $i$ along this trajectory equals
\[
\lim_{t\to \infty} \frac{1}{t}\int_0^t f_i(x.s)ds=(Ax^*)_i+b_i.
\]
Therefore,  $r_i(\mu)=\sum_j A_{ij}x_j^*+b_i$ for any ergodic measure $\mu$ supported by $\F(S)$. In particular,  assumption \textbf{A3b} is satisfied. 

\newcommand{\LVeq}{{\mathcal E}}
These observations imply that computing the invasion scheme and graph involves three steps. 
\begin{description}
\item[Step 1] Find the set $\LVeq$ of all feasible equilibria with at least one missing species: for each proper subset $S \subset [n]$, solve for $x \in \R^n$ such that $(Ax)_i=-b_i$ and $x_i > 0$ for $i \in S$, and  $x_j = 0$ for $j \notin S$. By assumption, $\LVeq$ is a finite set. The vertices $\S$ of the invasion graph are given by $S$ such that $\F(S)\cap \LVeq\neq \emptyset$.
\item[Step 2] Compute the invasion scheme $(r_i(S))_{S\in \S, i\in [n]}$ where $r_i(S)=\sgn ((Ax)_i +b_i)$ with $x=\LVeq \cap \F(S)$. 
\item [Step 3] Compute the invasion graph by checking the directed edge condition for each pair of subcommunities in $\S$, i.e., there is a directed edge from $S$ to $T$ iff $r_j(S)>0$ for all $j\in T\setminus S$ and $r_j(S)<0$ for all $j \in S\setminus T$. 
\end{description}

Two examples of using this algorithm for different $5$ species competitive communities are shown in Figure~\ref{fig:LV}. \bb In the case of community in the top panel, we also plot the vertex $[n]$ and transitions to this community. \ee  For both examples, the invasion graph is acyclic.  For the community in the top panel, there is a unique $-i$ community for each species and species $i$ has positive invasion growth rates at  this  community. Hence, Corollary~\ref{cor1} implies that this system is robustly permanent, see sample simulation in the upper right panel of Figure~\ref{fig:LV}. \bb Three of $-i$ communities ($i=2,3,4$) are co-dimension one and, consequently, species $i$ invading these communities (green directed edges) results in all species coexisting. The other two $-i$ communities ($i=1,5$) have more missing species. For example, the $-1$ community $\{2,3,4\}$ also misses species $5$.    When species $1$ invades this community,  species $3$ and $4$ are displaced leading to the $-5$ community $\{1,2\}$.  Successive single species invasions by species $5$, $3$ (or $4$), and then $4$ (or $3$) assemble the full community.  \ee For the community in the lower panel of Figure~\ref{fig:LV}, there are nine  $-i$ communities. For three of these $-i$ communities, species $i$ has negative invasion growth rates. Hence, the system isn't permanent. \bb Two of these $-i$ communities ($\{1,2,5\}$ and $\{1,3,4\}$) correspond to permanent subsystems where all the missing species have negative invasion growth rates. Hence, these $-i$ communities correspond to attractors for the full model dynamics. Moreover, each is a $-i$ community for each of the missing species e.g. $\{1,3,4\}$ is a $-2$ and $-5$ community. The third of these uninvadable $-i$ communities ($\{1,3,4,5\}$) is a non-permanent system due to the attractor on the boundary for the $\{1,3,4\}$ community. This explains the directed edge from $\{1,3,4,5\}$ to $\{1,3,4\}$. \ee The lower, right hand panel of Figure~\ref{fig:LV} demonstrates that the dynamics approach a three species attractor corresponding to one of the $-i$ communities.

Certain modifications of the classical Lotka-Volterra equations also satisfy the time-averaging property. Hence, for these modifications, computing the invasion scheme and the invasion graph also reduces to solving systems of linear equations. For example, \citet{jmb-18} showed that if the intrinsic rates of growth are driven by a uniquely ergodic process (e.g. periodic, quasi-periodic), then this reduction is possible. In this case, one uses auxiliary variables $\frac{dy}{dt}=f(y)$ that are uniquely ergodic and replace $b$ with vector valued functions $b(y)$. 

\subsection{Competitors sharing a switching predator}\label{sec:switch-pred}

Theoretical and empirical studies have shown that predators can mediate coexistence between competing prey species~\citep{paine1966,hutson_vickers1983,kirlinger1986,schreiber1997}. For example, \citet{hutson_vickers1983,schreiber1997} showed that a  generalist predator with a type I or II functional response can mediate coexistence when one prey excludes the other, but can not mediate coexistence when the prey are bistable, i.e., the single prey equilibria are stable in the absence of the predator. Here, we re-examine these conclusions by considering a modified Lotka-Volterra model accounting for predator switching~\citep{kondoh2003}. 

Let $x_1,x_2$ be the densities of two prey competitors. Let $x_3$ be the density of a predator whose prey preference is determined by the relative densities of the two prey species. Specifically, if  $y$ is the fraction of predators actively searching for prey $1$ and $1-y$ is the fraction actively searching for prey $2$, then we assume predator's switch between prey at a rate proportional to the prey densities, i.e., $\frac{dy}{dt}=x_1(1-y)-x_2y$. For simplicity, we assume the two prey species have a common intrinsic rate of growth $r$, normalized intraspecific competition coefficients and a common interspecific competition coefficient $\alpha$. Under these assumptions, the predator-prey dynamics are 
\begin{equation}\label{eq:apparent}
\begin{aligned}
\frac{dx_1}{dt}&=rx_1(1-x_1-\alpha x_2)-ax_1 y x_3 \\
\frac{dx_2}{dt}&=rx_2(1-x_2-\alpha x_1)-ax_2 (1-y)x_3 \\
\frac{dx_3}{dt}&=ax_1yx_3+ax_2 (1-y)x_3 -d x_3\\
\frac{dy}{dt}&=x_1(1-y)-x_2y
\end{aligned}
\end{equation}
where $a$ is the attack rate of the predator and $d$ is the per-capita death rate of the predator. The state space is $\K=[0,\infty)^3\times [0,1].$

As we are interested in predator mediated coexistence, we assume that $a>d$ to ensure the predator always persists. Under this assumption, the single prey-predator subsystem $x_i-x_3-y$ with $i=1,2$ has a unique globally stable equilibrium given by $x_i=d/a$, $x_3=r(1-d/a)/a$ and $y=1,0$ for $i=1,2$, respectively. If $\alpha\in[0,1)$, then the $x_1-x_2-y$ prey subsystem has a globally stable equilibrium $x_1=x_2=1/(1+\alpha)$ and $y=1/2$.  Alternatively, if $\alpha >1$, then the $x_1-x_2-y$ prey subsystem is bistable with a saddle at $x_1=x_2=1/(1+\alpha)$ and $y=1/2$. As the invasion graphs for $\alpha\neq 0$ are acyclic, Theorem~\ref{thm:main} implies that robust permanence occurs if and only if the three equilibria associated with the  $-i$ communities $\{1,2\},\{1,3\},\{2,3\}$ are invadable. Invasibility of $\{1,2\}$ and $\{1,3\}$ requires that $r(1-\alpha d/a)>0$. This occurs whenever $a/d>\alpha$, i.e., predation is sufficiently strong relative to interspecific competition. Invasibility of $\{2,3\}$  requires $a/d>1+\alpha$. In particular, unlike the case of non-switching predators~\citep{hutson_vickers1983}, predator-mediated coexistence is possible in the case of bistable prey. It is worth noting that in this case, the invasion scheme is not sequentially permanent.

\begin{figure}
\includegraphics[width=0.45\textwidth]{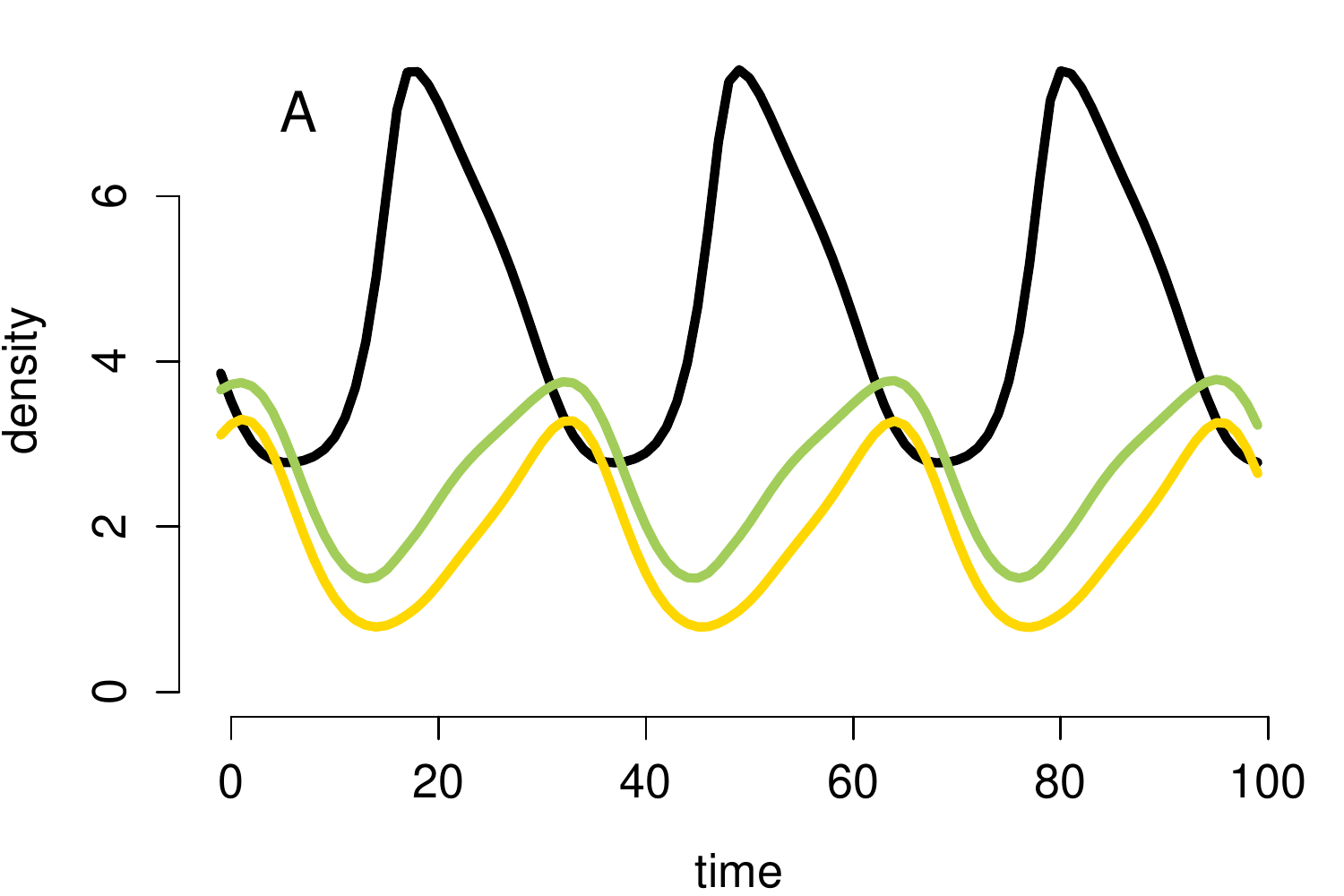} 
\includegraphics[width=0.45\textwidth]{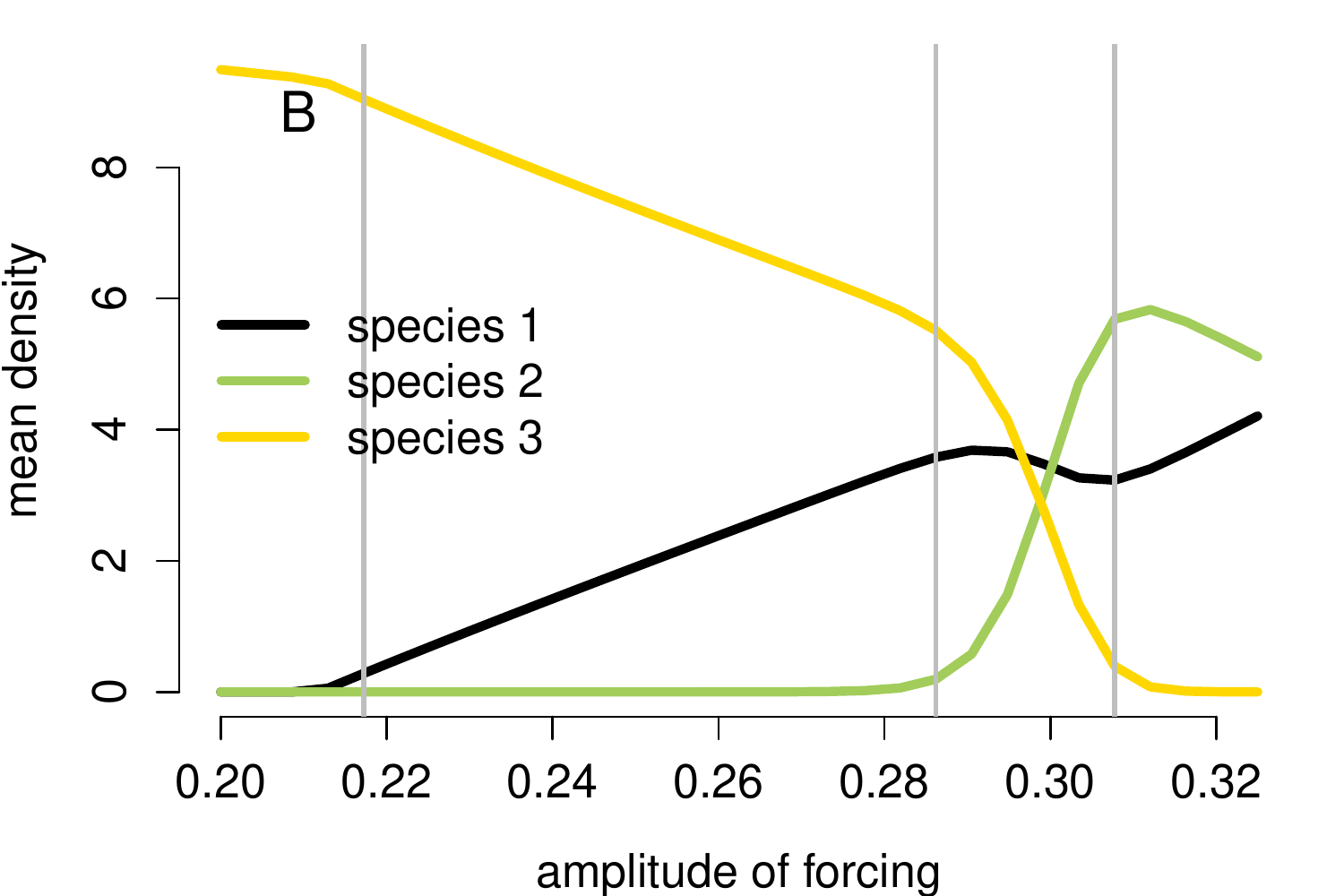} 
\caption{Periodic attractors for three competing species in a periodically-forced chemostat. In A, a periodic attractor at which all three species coexist (with $a=0.3$). In B, the mean densities of all three species along periodic attractors for increasing amplitude of the periodically-forced dilution rate.  Parameter values: $D_0=0.4675$, $\omega=0.2$, $\alpha_1=1,\alpha_2=0.7,\alpha_3=0.64$, $\beta_1=1,\beta_2=0.3,\beta_3=0.2$, $\bb R \ee_0=11$, and $a=0.3$ in A and as shown in B.}\label{fig:chemostat1}
\end{figure}

\begin{figure}
\includegraphics[width=0.22\textwidth]{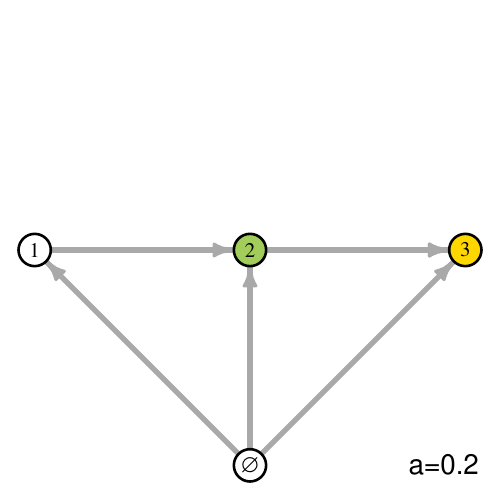}
\includegraphics[width=0.22\textwidth]{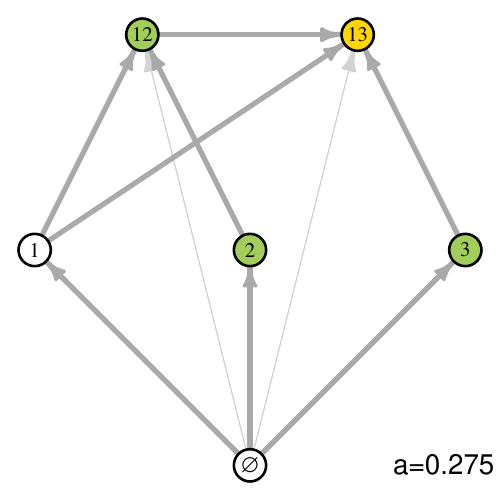}
\includegraphics[width=0.22\textwidth]{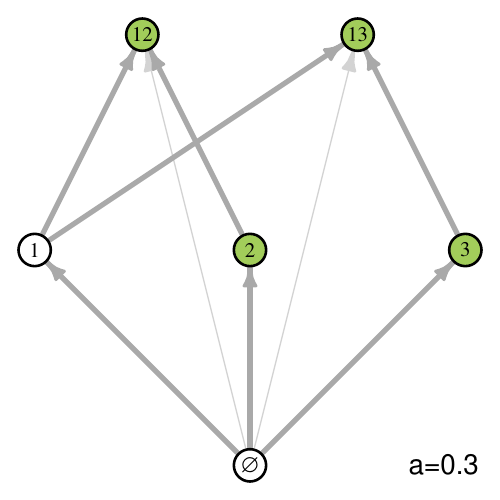}
\includegraphics[width=0.22 \textwidth]{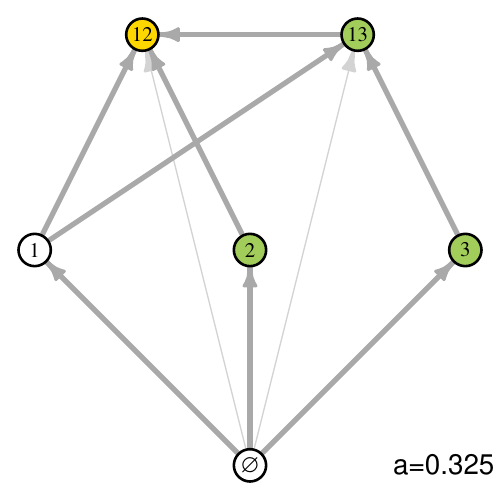}
\caption{Invasion graphs for three competing species in a periodically-forced chemostat for increasing values of the amplitude $a$ of the dilution rate. Shaded nodes correspond to  $-i$ subcommunities; olive green shading corresponds to a positive invasion growth rate of species $i$ and yellow shading a negative invasion growth rate. Parameters as in Figure~\ref{fig:chemostat1}}\label{fig:chemostat2}
\end{figure}

\subsection{Three competing species in a periodically forced chemostat} 

Chemostat are used in laboratories to study the dynamics of interacting microbial populations. As they provide a highly controlled environment, they are the basis of many mathematical modeling studies~\citep{smith_waltman1995}. For example, when species compete for  a single limiting resource with a constant inflow, \citet{hsu_hubbell1977} proved that (generically) the species with the lowest break-even point excludes all others. When the resource inflow, however, fluctuates, coexistence is possible. For example, using a mixture of numerics and analysis, \citet{lenas_pavlou1995} and \citet{wolkowicz1998} showed that three competing species can coexist when the inflow rate varies periodically. Specifically, if $\bb R \ee$ denotes the density of the resource in the chemostat and $x_i$ the density of competitor $i$, then they considered a model of the form:
\begin{equation}\label{chemostat}
\begin{aligned}
\frac{d\bb R \ee}{dt}=&(\bb R \ee_0-\bb R \ee)D(t)-\sum_{i=1}^3f_i(\bb R \ee)x_i\\
\frac{dx_i}{dt}=&x_i(t)(f_i(\bb R \ee)-D(t)) \quad i=1,2,3
\end{aligned}
\end{equation}
where \bb $R_0$ is the incoming resource concentration, \ee $D(t)=D_0+a\cos(\omega t)$ is a periodically fluctuating dilution rate, and $f_i(\bb R \ee)=\frac{\alpha_i \bb R \ee}{\beta_i+\bb R \ee}$ corresponds to a type II functional response. We  can put this model into our coordinate system by defining $y_1=\bb R \ee$ and $(y_2,y_3)$ to be points on the unit circle: 
\begin{equation}\label{chemostat2}
\begin{aligned}
\frac{dy_1}{dt}&=(\bb R \ee_0-y_1)D(y_2)-\sum_{i=1}^3f_i(y_1)x_i\\
\frac{dx_i}{dt}&=x_i(t)(f_i(y_1)-D(y_2)) \qquad i=1,2,3\\
\frac{dy_2}{dt}&=-\omega y_3\quad
\frac{dy_3}{dt}=\omega y_2 \mbox{ with }y_2^2+y_3^2=1
\end{aligned}
\end{equation}
where $D(y_2)=D_0+a\,y_2.$ The state space for \eqref{chemostat2} is $\C=\R_+^4\times S^1$ where $S^1\subset \R^2$ denotes the unit circle.

Using a numerically based invasion analysis, \citet{wolkowicz1998} showed that  \eqref{chemostat2} is permanent for the parameter values  $\alpha_1=1,\alpha_2=0.7,\alpha_3=0.64$, $\beta_1=1,\beta_2=-.3,\beta_3=0.2$,$D_0=0.4675$, $\omega=0.2$, $a=0.3$, and $\bb R \ee_0=11.$ Figure~\ref{fig:chemostat1}A plots the time series for what appears to be a global periodic attractor at which the three species coexist. Varying the amplitude of the dilution rate, however, can lead to the loss of one or two species (Figure~\ref{fig:chemostat1}B): at amplitudes higher than $0.3$, species $3$ is lost; at slightly lower amplitudes than $0.3$, species $2$ is lost; at much lower amplitudes, both species $1$ and $2$ are lost. 

To better understand these effects of the amplitude of fluctuations on species coexistence, we numerically calculated the Lyapunov exponents for all subsystems and created the invasion graphs for different amplitudes of the dilution rate (Figure~\ref{fig:chemostat2}). At the amplitude value $a=0.3$ used by \citet{wolkowicz1998}, we recover the invasion graph suggested by their analysis. Specifically there are only two -$i$ communities, $\{1,2\}$, $\{1,3\}$, and both of these communities can be invaded by the missing species. Hence, Theorem~\ref{thm:main} implies robust permanence. At a higher amplitude of $a=0.325$, the community $\{1,2\}$ can no longer be invaded by species $3$ and permanence no longer occurs, consistent with the loss of species $3$ in Figure~\ref{fig:chemostat1}B at $a=0.325$. At a lower value of the amplitude, $a=0.275$, the $-2$ community is not invadable, a prediction consistent with species $2$ being excluded in Figure~\ref{fig:chemostat1}B at $a=0.275$. At an even lower value of the amplitude, $a=0.2$, the invasion graph in Figure~\ref{fig:chemostat2} dramatically changes with the community determined by species $3$ resisting invasion from the other two species, consistent with only species $3$ persisting in Figure~\ref{fig:chemostat1} at $a=0.2$. 

\section{Discussion~\label{sec5}}

Modern coexistence theory (MCT) \bb decomposes and compares  invasion growth rates to identify mechanisms of coexistence~\citep{chesson-94,chesson-00b,letten2017,Chesson_2018,barabas2018,ellner_snyder2018,grainger_levine2019,grainger2019,godwin2020,chesson2020}.  Our work addresses two key question for this theory: When are signs of invasion growth rates sufficient to determine coexistence? When signs are sufficient, which positive invasion growth rates are critical for coexistence? \ee

To answer these questions, we introduced invasion schemes and graphs. The invasion scheme catalogs all invasion growth rates associated with every community missing at least one species. The invasion graph describes potential transitions between communities using the invasion growth rates. Potential transition from a community $S$ to a community $T$ occurs if (i) all the species in $T$ but not in $S$ have positive invasion growth rates when $S$ is the resident community and (ii) all the species in $S$ but not in $T$ have negative invasion growth rates when $T$ is the resident community.  Our definition of invasion graphs is related to what is often called an assembly graph in the community assembly literature~\citep{post-pimm-83,law-morton-96,morton1996,servan2021}. For example, \citet[page 1030]{servan2021} define  assembly graphs for Lotka-Volterra systems. Like our definition applied to Lotka-Volterra systems (see section \ref{sec:LV}), vertices correspond to feasible equilibria of the model. Unlike our definition, \citet{servan2021}  only consider transitions between communities due to single species invasions.  This more restrictive definition, however, may not exclude heteroclinic cycles between equilibria due to multiple species invasion attempts, i.e., the ``1066 effect'' of  \citet{lockwood1997}. These heteroclinic cycles may exclude the possibility of determining permanence only based on the signs of the invasion growth rates~\citep{hofbauer-94}. 

We show that the signs of the invasion growth rates determine coexistence whenever the invasion graph is acyclic, i.e., there is no sequence of invasions starting and ending at the same community. For acyclic graphs, we identify a precise notion of what \citet{chesson-94} has called ``$-i$ communities'', i.e., the communities determined in the absence of species $i$. Specifically, these are communities where (i) species $i$ is missing, and (ii) all other missing species have a negative invasion growth rate. $-i$ communities can be found, approximately (see Remark~\ref{remark2}), by simulating initial conditions supporting all species but species $i$ for a sufficiently long time, removing ``atto-foxes''~\citep{sari2015migrations,fowler2021atto}, and seeing what species are left. This characterization ensures that each species $i$ has at least one $-i$ community associated with it.  

When the invasion graph is acyclic, we show that robust permanence occurs if, and only if, at each  $-i$ community, species $i$ has a positive invasion growth rate. Thus,  this result helps define  the domain of modern coexistence theory which \bb relies \ee on the signs of invasion growth rates \bb determining \ee coexistence~\citep{macarthur_levins1967, chesson-94,chesson-00b,letten2017,barabas2018,ellner_snyder2018,grainger_levine2019,grainger2019,godwin2020,chesson2020}. 

\bb Our work also highlights the importance of going beyond  average Lyapunov functions when only using qualitative information about invasion growth rates. The average Lyapunov function condition for permanence  requires the existence of positive weights $p_i$ such that  $\sum_i p_i r_i(\mu)>0$  all ergodic measures $\mu$ supporting a strict subset of species \citep{hofbauer-81}. This sufficient condition for permanence has received more attention in the theoretical ecology literature~\citep{law-blackford-92,law-morton-93,law-morton-96,morton-law-pimm-drake-96,Chesson_2018,chesson2020} than  sufficient topological conditions using Morse decompositions~\citep{garay-89,hofbauer-so-89}, or conditions using invasion growth rates with Morse decompositions~\citep{jde-00,garay-hofbauer-03,nonlinearity-04,jde-10,roth-etal-17,jmb-18}. This is likely due to the more technical nature of these latter papers. However, when one only knows the signs of the invasion growth rates, the average Lyapunov condition only works for specific types of acyclic graphs. Specifically, those graphs corresponding to a nested sequence of permanent communities, $\{1\}, \{1,2\},\{1,2,3\},\dots, \{1,2,3,\dots,n\}$, where species $i+1$ has non-negative invasion growth rates for all communities including species $1,2,\dots,i$. While these special graphs arise in some situations (e.g.  diffusive competition~\citep{Chesson_2018,jmaa-02} or certain generalizations of mutual invasibility~\citep{chesson-kuang-08}), many communities do not exhibit this special structure.  \ee

There remain many mathematical challenges for an invasion-based approach to coexistence. First, while our main assumption \textbf{A3b} naturally holds for Lotka-Volterra and replicator systems, they are too strong for many other systems. Notably, our assumptions do not allow for the invariant sets supporting multiple ergodic measure at which the invasion growth rates for a species have opposite sign. What can be done in these cases is not clear as they can cause complex dynamical phenomena such as riddled basins of attraction~\citep{alexander-kan-yorke-you-92,hhrs-04} and open sets of models where permanence and attractors of extinction are intricately intermingled~\citep{nonlinearity-04}. More optimistically, for stochastic models accounting for environmental stochasticity, the story may be simpler. For these models, permanence corresponds to stochastic persistence -- a statistical tendency of all species staying away from low densities~\citep{chesson-82,benaim-etal-08,jmb-11,benaim2018stochastic,hening-nguyen-18,benaim_schreiber2019}. Under certain natural irreducibility assumptions~\citep{jmb-11,hening-nguyen-18,hening2020classification}, each face $\F(S)$ supports at most one ergodic measure; \textbf{A3b} naturally holds for these models. Using the stochastic analog of invasion growth rates, one can define invasion schemes and invasion graphs as we do here. For these models, it is natural to conjecture: if the invasion graph is acyclic and all $-i$ communities are invadable, then the model is stochastically persistent. 

Dealing with cyclic invasion graphs is another major mathematical challenge. When these cycles are sufficiently simple, their stability properties can be understood using either average Lyapunov functions or %geometric approaches
Poincar\'e return maps~\citep{hofbauer-94,krupa-melbourne-95,krupa-97}. For more complex heteroclinic cycles (even between equilibria), the path forward for characterizing coexistence via invasion growth rates is less clear~\citep{hofbauer-94,brannath-94}. Even for cyclic graphs where invasion growth rates characterize coexistence, it remains unclear how to carry out the second step of  modern coexistence theory, i.e.,  how best to decompose and compare invasion growth rates to identify the relative contributions of different coexistence mechanisms.  We hope that future mathematical advances on these issues will be incorporated into a next version of the modern coexistence theory (MCT v2.1).

\textbf{Acknowledgements}: We thank Adam Clark and J\"{u}rg Spaak for providing useful feedback on an early draft of the manuscript, and Peter Chesson and two anonymous reviewers for  their  thoughtful reviews that further improved this work.  SJS was supported in part by U.S. National Science Foundation grants DMS-1716803, DMS-1313418, and a Simons Visiting Professorship to SJS sponsored by %JH and 
Reinhard B\"{u}rger. 

\bibliography{IG}

\setcounter{section}{0}
\renewcommand{\thesection}{Appendix \Alph{section}}

\section{Proof of Theorem~\ref{thm:main}}\label{sec:proof1}

Throughout this proof, we assume that assumptions  \textbf{A1}--\textbf{A3} hold. As defined earlier, let $\S$ be the set of all subsets $S$ of $[n]$ 
such that $S = S(\mu)$ for some ergodic invariant measure $\mu$ on $\C_0$
i.e. the set of subcommunities. For any $z=(x,y)$, define $\pi_i z =x_i$. Recall that the $\alpha$-limit and $\omega$-limit sets of a point $z\in \K$ are given by $\alpha(z)=\{z':$ there a sequence $t_k\downarrow -\infty$ such that $\lim_{k\to\infty} z.t_k=z'\}$ and $\omega(z)=\{z':$ there a sequence $t_k\uparrow +\infty$ such that $\lim_{k\to\infty} z.t_k=z'\}$.

The key lemma for the proof is the following:

\begin{lemma} Let $S, T \in \S$ be two subcommunities. 
If there exists $z\in \Gamma_0$ such that $\alpha(z)\subset \F(S)$ and $\omega(z)\subset \F(T)$, then $S\rightarrow T$ in the invasion graph.
\end{lemma}

\begin{proof} As $\alpha(z)\subset \F(S)$ and $\omega(z)\subset \F(T)$, we have that $\pi_\ell z>0$ for all $\ell \in S\cup T$.

First, we show that $r_\ell(S)>0$ for all $\ell\in T\setminus S$. For all $t>0$, let $\eta^-_t$ be the probability measure defined by $\int_\C h(z')\eta^-_t(dz')=\frac{1}{t}\int_0^t h(z.(-s))ds$ for any continuous function $h:\C\to \R$. Since $\alpha(z)\subset \F(S)$, there exists a sequence of times $t_k\uparrow \infty$ and a probability measure $\eta^-$  satisfying $\eta^-(\F(S))=1$ such that $\eta^-_{t_k}$ converges in the weak* topology as $k\uparrow\infty$. Furthermore, the classical argument of the Krylov-Bogolyubov theorem~\citet{kryloff1937theorie} implies $\eta^-$ is invariant.

We have that $\log\frac{\pi_\ell z.(-t)}{\pi_\ell z}=-\int_0^tf_{\ell}(z.(-s))ds$ for
all $\ell\in S\cup T$ and $t>0$. As $\alpha(z)\subset \F(S)$,
\[
r_\ell(\eta^-)=\lim_{k\to\infty}\frac{1}{t_k}\int_0^{t_k}f_\ell(z.(-s))ds=-\lim_{k\to\infty}\frac{1}{t_k}\log \frac{\pi_\ell z.(-t_k)}{\pi_\ell z }\ge 0 
\]
for all $\ell \in S\cup T$. By the ergodic decomposition theorem, for each $\ell \in T \setminus S$, there exists an ergodic measure $\mu$ with $S(\mu)=S$ and $r_\ell(\mu)\ge 0$. Hence, assumption \textbf{A3a} implies that $r_\ell(S)>0$ for all $\ell\in T \setminus S$.

Second, we can use a similar argument to show that $r_\ell(T)<0$ for all $\ell \in S\setminus T$. In this case, we use the forward empirical measures $\eta_t^+$ defined by  $\int_\C h(z')\eta^+_t(dz')=\frac{1}{t}\int_0^t h(z.s)ds$ for any continuous function $h:\C\to \R$.
\end{proof}

Next, we construct a Morse decomposition of $\Gamma_0$ determined by the invasion graph. Recall, a finite collection of compact, isolated invariant sets $\{M_1,\dots,M_k\}$ with $M_i\subset \Gamma_0$ is a Morse decomposition of $\Gamma_0$ if for all $z\in \Gamma_0 \setminus \cup_i M_i$ there exist $j>i$ such that $\alpha(z)\subset M_i$ and $\omega(z)\subset M_j$.

\begin{lemma} 
Let $k \in \{ 0,1,\dots , n-1\}$, and $\S_k =\{I \in \S: \abs{I} \leq k\}$.
Suppose $\IG$ is acyclic. 
Then: 

For each $I \in \S_k$ there is a 
nonempty compact invariant subset $M_I  \subset \F(I)$ such that
\begin{itemize}
    \item[1.] $\omega(z)\subset \bigcup_{I \in \S_k} M_I$ for all $z\in \bigcup_{I \in \S_k} \F(I)$.
    \item[2.] For each bounded complete solution $z.t \in \bigcup_{I \in \S_k} \F(I)$, $\alpha(x)\subset \bigcup_{I \in \S_k} M_I$.
%this follows from the other two points by Garay's Lemma: each isolated acyclic covering is a Morse decomposition
    \item[3.]  Each $M_I$ is isolated in $\K$.
    %oder: Each chain recurrent point in  $\partial \R^n_0$ is in some $M_I$.
     \item[4.] The family of invariant sets $\{ M_I: I \in \S_k \}$ is a Morse decomposition of $\bigcup_{I \in \S_k} \F(I) \cap \Gamma_0$.
\end{itemize}
\end{lemma}

\begin{proof} We prove this lemma by induction on $k$. % = \abs{I}$ %(size of subcommunity)
%, for the $k$-skeleton $\{ I \in \S: \abs{I} \leq k\}$, 
%for $k = 0,1, \dots, n-1$.

If $k = 0$, then $M_{\emptyset} = \{0\}$ is hyperbolic, and, consequently,  isolated in $\C$. Properties 1, 2, and 4, hold immediately.

%$k = 1$: $M_{\{i\}}$ is a (possibly empty) interval contained %(or is a compact set of equilibria)
% in the positive $x_i$-axis. 
% 3 possibilites

%$k = 2$: $M_{\{i,j\}}$ is a (possibly empty) compact invariant set in $\F(\{i,j\}) \subseteq x_i$-$x_j$ plane

Suppose the Lemma holds for $k-1$, and all the $M_I$ for $\abs{I} < k$ are given. Then we define $M_I$ for $\abs{I} = k$ as the maximal compact invariant subset in $\F(I)$. It exists (i.e.\ there are no invariant sets arbitrarily close to the   boundary of $\F(I)$) as the 
family of  $M_J$ with $J \subsetneq I$ forms a Morse decomposition of the boundary $\partial \F(I)$ of $\F(I)$, each $M_J$ is isolated, and hence $\partial \F(I)$ is isolated. %are all isolated, and are acyclic.
Hence properties 1. and 2. hold.

To show property 3., suppose to the contrary that $M_I$  is not isolated in $\C$. Then for every $\eps > 0$ there is a 
$z^\eps\in \C$ such that  $\mbox{dist}(z^\eps.t, M_I) < \eps$ but $z^\eps.t \notin \F(I)$ for all $t \in \R$. Hence there is a $j \notin I$ with $\pi_j z^\eps.t > 0$ for all $t\in \R$. As in the proof of Lemma 2, we can find invariant measures $\mu^\eps_+, \mu^\eps_-$ with $r_j( \mu^\eps_+) \le 0$ and $r_j( \mu^\eps_-) \ge 0$. Passing to appropriate subsequences $\eps_m \to 0$, these measures converge weak$^*$, $ \mu^{\eps_m}_+ \to \mu_+ ,  \mu^{\eps_m}_- \to \mu_-$ to (not necessarily ergodic) invariant measures supported on $\F(I)$. These invariant measures satisfy
$r_j( \mu_+) \le 0$ and $r_j( \mu_-) \ge 0$ that contradicts assumption {\bf A3}. Thus, property 3. holds. 

Finally, the assumption that $\IG$ is acyclic implies property 4 by choosing a suitable order on $\S_k$.
\end{proof}

Taking $k =n-1$ we get

\begin{lemma} 
If $\IG$ is acyclic then the  
family of invariant sets $\{ M_I: I \in \S \}$ is a Morse decomposition of $\Gamma_0$.
\end{lemma}

\bb We also need the following lemma which follows from Assumption \textbf{A3} and \citet[Corollary 1]{jde-98} (see, also, \citet[Exercise I.8.5]{mane-83}).
\begin{lemma}
Let $I\in \S$ and $i\in [n]\setminus I$ be such that $r_i(I)>0$. Then there exists $\tau>0$ and $\alpha>0$ such that 
\[
\frac{1}{\tau}\int_0^\tau f_i(z.t)dt \ge \alpha
\]

for all $z\in M_I.$
\end{lemma}

Assume that for every $I$, there is $j$ such that $r_j(I)>0$. Then Lemmas 4 and 5 and \citep[Theorem 2]{jmb-18} imply \eqref{eq:main} is robustly permanent. 

\ee

%Permanence follows from the Corollary of Theorem 2 by \citet{garay-89}. Robust permanence follows from \citet{jde-00} or %Theorem 5.5 in 
%\citet{garay-hofbauer-03}.

\section{\bb Proof of Proposition~\ref{prop:IS}\ee}\label{sec:proof2}
\bb
Assume $\IS$ is sequentially permanent. Without loss of generality, we can assume the sequence is $\{1,2,\dots,n\}$. Let $C$ be a matrix with $\sgn(C)=\IG$. Let $c_+>0$ and $-c_-<0$, respectively, be the minimum of the positive and negative elements of $C$. By the definition of sequential permanence, $C_i([n-1])=r_i([n-1])=0$ for $i\le n-1$ and $\sgn(C_n([n-1]))=r_n([n-1])=+1$. Hence, for any choice of $p\gg 0$, $\sum_i p_i C_i([n-1])>0$. Choose $p_n=1$. Next, for any $S\in \S$ such that $\{1,\dots, n-2\}\subset S$,  $C_i(S)=0$ for $i\le n-2$, $C_{n-1}(S)\ge c_+$, and $C_{n}(S)\ge -c_-$. Hence, for any choice of $p=(p_1,\dots,p_{n-1},1)\gg 0$, $\sum_i p_i C_i(S)\ge p_{n-1}c_+-c_-$. Choose $p_{n-1}=2c_-/c_+$. Then, $\sum_i p_i r_i(S)>0$ for any $S\in \S$ such that $\{1,\dots, n-2\}\subset S$. Proceeding inductively (i.e. choosing $p_k\ge 2(p_{k+1}+\dots +p_{n-1}+1)c_-/c_+)$, we get $p=(p_1,\dots,p_k)\gg 0$ such that $Cp\gg 0.$

Now suppose $\IG$ is not sequentially permanent. We begin by assuming this failure of sequential permanence occurs at the first step i.e. there is no species such that $r_i(S)\ge 0$ for all $S\in \S.$ Then for each species $i$, there is  $S_i\in \S$ such that $r_i(S_i)<0$. Let $C$ be such that its positive entries equal  $1/n$,  its negative entries equal $-1$, and $\sgn(C)=\IS$.  Suppose, to the contrary, there exists $p=(p_1,\dots,p_n)\gg 0$ such that $Cp\gg 0$.  Then $0<\sum_i p_i C_i(S_j)\le -p_j + \sum_{i\neq j} p_i/n $ for any $j\in [n]$. Adding these $n$ inequalities leads to a contradiction. 
This completes the proof when sequential permanence fails at the first step. Now suppose that the definition fails at the $k+1$ step with $k<n$ i.e. (i) there exist distinct species $\ell_1,\ell_2,\dots,\ell_k$ such that column $\ell_i$ of $\IS_i$ has only non-negative entries, and (ii) every column of $\IS_{k+1}$ has at least one negative entry. Then we can use the same argument restricted to $\IS_{k+1}$ as $r_i(S)=0$ for all $i\le k$ and $S\in \S$ satisfying $[k]\subset S.$
\ee 

\section{Parameter values for Figure~\ref{fig:LV}\label{sec:parms}}
Lotka-Volterra parameters for the coexistence panels in Figure~\ref{fig:LV}
\[
A=\begin{pmatrix}
-1.02 & -0.71 & -0.84 & -0.47 & -1.42 \\ 
  -0.43 & -2.19 & -0.35 & -0.73 & -1.03 \\ 
 -1.31 & -1.10 & -1.83 & -0.59 & -0.74 \\ 
  -1.20 & -0.51 & -0.01 & -1.99 & -0.67 \\ 
 -0.47 & -1.00 & -1.23 & -1.11 & -1.42 \\ 
\end{pmatrix}\qquad
b=\begin{pmatrix}
1.00 \\ 
  1.00 \\ 
  1.00 \\ 
  1.00 \\ 
  1.00 \\ 
  \end{pmatrix}
\]
Lotka-Volterra parameters for the non-coexistence panels in Figure~\ref{fig:LV}
\[
A=\begin{pmatrix}
-2.23 & -0.86 & -0.96 & -0.18 & -0.04 \\ 
  -1.07 & -1.58 & -1.13 & -0.79 & -0.20 \\ 
  -1.45 & -0.77 & -1.14 & -0.06 & -1.23 \\ 
  -0.12 & -0.55 & -0.46 & -2.41 & -1.35 \\ 
  -0.08 & -0.33 & -1.49 & -0.10 & -1.14 \\ 
\end{pmatrix}\qquad
b=\begin{pmatrix}
1.00 \\ 
  1.00 \\ 
  1.00 \\ 
  1.00 \\ 
  1.00 \\ 
  \end{pmatrix}
\]

\end{document}